\documentclass[sigconf]{acmart}

\usepackage{balance} 

\usepackage{times}

\usepackage{amsmath}
\usepackage{dutchcal}
\usepackage{subfigure}
\usepackage[justification=centering]{caption}

\usepackage{filecontents}

\usepackage{times}

\usepackage{psfrag}
\usepackage{epsfig}
\usepackage{graphicx,subfigure}
\usepackage{color}
\usepackage{bm}
\usepackage{stfloats}
\usepackage{epstopdf}
\usepackage{verbatim}
\usepackage{graphicx}   
\usepackage{enumitem}

\usepackage{textcomp}

\renewcommand{\textrightarrow}{$\rightarrow$}

\usepackage{amsthm}
\usepackage{tablefootnote}
\usepackage{latexsym,verbatim}
\usepackage{psfrag,epsfig}
\usepackage{color}
\usepackage{epstopdf}
\usepackage{listings}
\usepackage{stfloats}
\usepackage{tikz}
\usepackage{url}
\usepackage{algorithmic}
\usepackage{booktabs}

\newcommand{\name}{CrowdMine}

\newtheorem{theorem}{Theorem}

\newtheorem{constraint}{Constraint}

\usepackage[linesnumbered,ruled]{algorithm2e}
\SetAlFnt{\footnotesize} 
\SetKwBlock{KwFunction}{function}{}
\SetKwBlock{KwOn}{on}{}
\SetKwBlock{KwIf}{if}{}
\SetKwBlock{KwContinue}{continue}{}
\SetKwBlock{KwPeriodically}{periodically}{}
\SetKwBlock{KwState}{state}{}
\SetKwBlock{KwAtomic}{atomically}{}
\SetKwBlock{KwConcurrently}{concurrently}{}
\SetCommentSty{textup} 
\SetKwComment{KwIttayComment}{(}{)} 
\SetKw{KwAnd}{and}
\SetKw{KwOr}{or}
\SetKw{KwSetTimer}{setTimer}
\SetKwBlock{KwExpTimer}{on timer expiration}{}
\SetKwFor{ForAll}{forall}{}

\newenvironment{proof-sketch}{\noindent{\bf Sketch of Proof:}\hspace*{1em}}{\qed\bigskip}

\AtBeginDocument{%
	\providecommand\BibTeX{{%
			\normalfont B\kern-0.5em{\scshape i\kern-0.25em b}\kern-0.8em\TeX}}}

\setcopyright{none}
\renewcommand\footnotetextcopyrightpermission[1]{}

\acmConference[Conference acronym 'XX]{Make sure to enter the correct
	conference title from your rights confirmation emai}{June 03--05,
	2018}{Woodstock, NY}
\acmPrice{15.00}
\acmISBN{978-1-4503-XXXX-X/18/06}

\begin{document}

\title{Crowdsourcing Work as Mining: A Decentralized Computation and Storage Paradigm}
	\author{Canhui Chen*, Zerui Cheng*, Shutong Qu, and Zhixuan Fang}
    \thanks{*Equal contribution.}
    \thanks{Canhui Chen, Zerui Cheng, Shutong Qu are with Tsinghua University, Beijing, China. Zhixuan Fang is with Tsinghua University, Beijing, China and Shanghai Qizhi Institute, Shanghai, China. }
    \thanks{Corresponding author: Zhixuan Fang. Email: zfang@mail.tsinghua.edu.cn}
    







\begin{abstract}
Proof-of-Work (PoW) consensus mechanism is popular among current blockchain systems, which leads to an increasing concern about the tremendous waste of energy due to massive meaningless computation. To address this issue, we propose a novel and energy-efficient blockchain system, CrowdMine, which exploits useful crowdsourcing computation to achieve decentralized consensus. CrowdMine solves user-proposed computing tasks and utilizes the computation committed to the task solving process to secure decentralized on-chain storage. With our designed ``Proof of Crowdsourcing Work'' (PoCW) protocol, our system provides an efficient paradigm for computation and storage in a trustless and decentralized environment. We further show that the system can defend against potential attacks on blockchain, including the short-term 51\% attack, the problem-constructing attack, and the solution-stealing attack. We also implement the system with 40 distributed nodes to demonstrate its performance and robustness. To the best of our knowledge, this is the first system that enables decentralized Proof of Useful Work (PoUW) with general user-proposed tasks posted in a permissionless and trustless network. 
\end{abstract}




%
%
%
%
%
\pagestyle{fancy}
\fancyhead{}

\maketitle

\section{Introduction}\label{sec:introduction}

Blockchain, with Bitcoin \cite{nakamoto2008bitcoin} and Ethereum \cite{buterin2014next} as its prevalent representatives, is a decentralized ledger that commits consensus and trust among distributed participants.
To coordinate the decentralized record generation, most blockchain systems adopt the Proof of Work (PoW) mechanism (e.g., Bitcoin).
In fact, according to  \cite{gervais2016security}, more than 90\% of blockchain-based cryptocurrencies are using PoW to achieve the decentralized consensus and security guarantees.
Specifically, PoW evaluates miners' computing power in solving random cryptographic problems and chooses the next block creator in proportion to their computing power.  Although the computation results are useless, the computational effort from miners is in general considered to be a guarantee for security.

The tremendous energy consumption of the PoW mining is widely recognized as a fundamental drawback of the blockchain and becomes a key challenge for its future development \cite{de2018bitcoin,beck2018beyond}.
Rauchs et.al. \cite{rauchs2020cambridge} show that Bitcoin mining consumes around 119.87 TWh of electricity each year, which is more than the electricity consumption of a medium-sized country. Moreover, the enormous energy consumption is still rapidly increasing (\cite{gallersdorfer2020energy,de2018bitcoin,truby2018decarbonizing}). 
Although researchers have proposed many consensus algorithms in hope of substituting PoW, these new consensus algorithms usually have their own inherent weaknesses. As an example,  Proof-of-Stake (PoS, \cite{king2012ppcoin}) is considered to be energy-efficient, but concerns rise on centralization, fairness and reliability (e.g., \cite{haouari2022novel} \cite{nair2021evaluation} \cite{saad2021pos}, see Section \ref{sec:related} for further discussion).
More importantly, it is indeed the massive computation in PoW mining that builds a fence against malicious tampering. Such a solid security guarantee is particularly attractive in the decentralized environment.

To address the challenge of energy waste in blockchain, in this paper, we propose a new decentralized computing system, \emph{CrowdMine}.
The system is decentralized and permissionless.
The key idea of our paradigm is that we utilize miners' computation resources to solve actual computing tasks proposed by distributed problem proposers (i.e., users), instead of solving the meaningless hash puzzle in PoW.
More importantly, we also utilize this useful computation to secure the transactions in the blockchain.
We further show that \emph{CrowdMine} can efficiently defend against the short-term 51\% attack, the problem-constructing attack, and the solution-stealing attack, maintaining security in the long term.
We summarize the design principles of the system as follows.
\begin{itemize}[leftmargin=*]
	\item A miner is eligible to create a new block if and only if the miner successfully finishes a computation task (i.e., solves a ``problem''), either proposed by users or the system. Thus, we utilize miners'  resources for useful computation.
	\item We measure the miner's computational
	contribution by the reward of the task, and propose the \emph{Maximum Aggregated Value protocol} (MAV) to resolve forking branches.
	\item Eligible miners publish blocks according to the proposed \emph{Proof of Crowdsourcing Work} (PoCW) protocol. PoCW chooses miners with probabilities in proportion to their computational contribution to users' proposed tasks.
\end{itemize}

To the best of our knowledge, CrowdMine is the first paradigm that enables PoUW with decentralized and general task posting in a permissionless and trustless network. Specifically, our proposal is novel and unique in two aspects compared with existing systems.
First, there are existing studies on Proof-of-Useful-Work (PoUW) that tried to replace PoW with useful computation, but all these works require at least one of the following strong assumptions, including a single problem proposer \cite{ball2017proofs}\cite{ball2018proofs},  a trusted third party or hardware \cite{zhang2017rem}\cite{chen2017security}, or a centralized organizer \cite{halford2014gridcoin}.
Second, in previous work about on-chain crowdsourcing (e.g., \cite{goertzel2017singularitynet}\cite{li2018crowdbc}),  the computation effort on solving user problems usually does not contribute to the security of the system. The system security relies on some existing blockchains,  and the crowdsourcing work is mostly an application beyond the underlying blockchain. As a result, these systems still adopt other consensus protocols for security, which do not fully exploit the value of computation work. 

\section{Related Work}\label{sec:related}
Researchers have proposed many miner selection protocols without massive computation. 
Typical proposals include Proof of Stake (PoS) \cite{king2012ppcoin}, Delegated Proof of Stake (DPoS) \cite{larimer2014delegated}, Ripple Protocol Consensus Algorithm (RPCA) \cite{schwartz2014ripple}, AlgoRand \cite{gilad2017algorand}, Proof of Activity \cite{bentov2014proof}, Proof of Prestige \cite{krol2019proof},  Proof of Trust \cite{bahri2018trust}, and more. 
However, these alternative protocols usually have their inherent weaknesses, among which security is the key concern, since these ``non-computation-based'' protocols significantly lower the bar for an attacker to revert records (e.g., \cite{li2017securing}\cite{kiayias2017ouroboros}\cite{brown2019formal}).

Proof of Useful Work (PoUW) arises as another way to resolve the energy inefficiency problem, where miners' computation power can be of practical use, while still being a security assurance for consensus.
One way for PoUW is to convert practical problems into hash puzzles. In theory, it can be shown that hard problems such as OV, 3SUM, APSP (\cite{ball2017proofs, ball2018proofs}), TSP\cite{loe2018conquering} and Discrete Logarithm \cite{hastings2019short} can be solved by constructing hash puzzles in traditional PoW. 
However, the range of applicable problems is still rather limited,  and a centralized trusted server is required to construct the hash puzzles and collect the results.
Another way is to replace the hash puzzle calculation with other kinds of math problems. Its representatives include PrimeCoin (number theoretical PoW) \cite{king2013primecoin} and Cuckoo cycle (graph theoretical PoW) \cite{tromp2015cuckoo}. 
However, these problems are hardly practical. Moreover, GridCoin \cite{halford2014gridcoin} incentivizes miners to put their computation power into BOINC \cite{anderson2004boinc}, but the work verification process is centralized.  CoinAI \cite{baldominos2019coin} incorporates the mining process with neural network training, but it is vulnerable to possible attacks from  dishonest miners. With the intervention of trustful third-party (such as Intel SGX \cite{costan2016intel}), REM \cite{zhang2017rem}, and PoET \cite{chen2017security} can make the miners' computation useful, but finding a trusted third-party could be difficult in a decentralized system with permissionless participation.

Researchers have also proposed many blockchain-enabled crowdsourcing systems. Separ \cite{amiri2021separ} is a multi-platform crowdsourcing system that enforces regulations in a privacy-preserving manner. SingularityNET \cite{goertzel2017singularitynet} is a distributed AI service systems. CrowdBC \cite{li2018crowdbc} is a decentralized crowdsourcing platform, and \cite{ducree2020open} further puts forward an open platform strategy that combines a blockchain-endowed token economy for finding trust in a decentralized setting.
But most of the proposed systems are built on top of existing blockchains (e.g., through smart contracts or other methods) to organize the crowdsourcing \cite{szabo1997formalizing}, relying on the underlying blockchains to guarantee the system security, which are still PoW-based in many cases.

To the best of our knowledge, our work is the first one that achieves PoUW for general problem solving and decentralized problem posting, where the useful computation guarantees an even higher level of security compared with traditional PoW.
Compared with blockchain-enabled crowdsourcing  \cite{goertzel2017singularitynet} \cite{li2018crowdbc}, our paradigm is different in that we incorporate problem solving into the mining process without relying on an underlying PoW hash computing. 
Compared with previous works on PoUW (e.g., \cite{zhang2017rem}\cite{chen2017security}\cite{halford2014gridcoin}), our system does not require a trusted third party, including the task allocation and the proof verification.
Compared with systems in \cite{ball2017proofs} \cite{ball2018proofs} \cite{loe2018conquering} \cite{hastings2019short} \cite{king2013primecoin} \cite{tromp2015cuckoo}, our system can solve a wide range of problems as long as verification is simpler than computation.

\section{System Architecture}\label{sec:model}
Our system is designed to utilize the distributed computational resources to solve user-proposed computation tasks and maintain a decentralized ledger at the same time.
To achieve decentralized task solving and record consensus, we propose Proof of Crowdsourcing Work (PoCW) as the mining protocol.
The aggregated computing power of miners, and the reward of computing tasks paid by the users, can altogether provide a security guarantee for the system.
The system architecture is shown in Figure \ref{fig:system_model}.

\begin{figure}[!ht]
	\centering
	\includegraphics[width=1\linewidth]{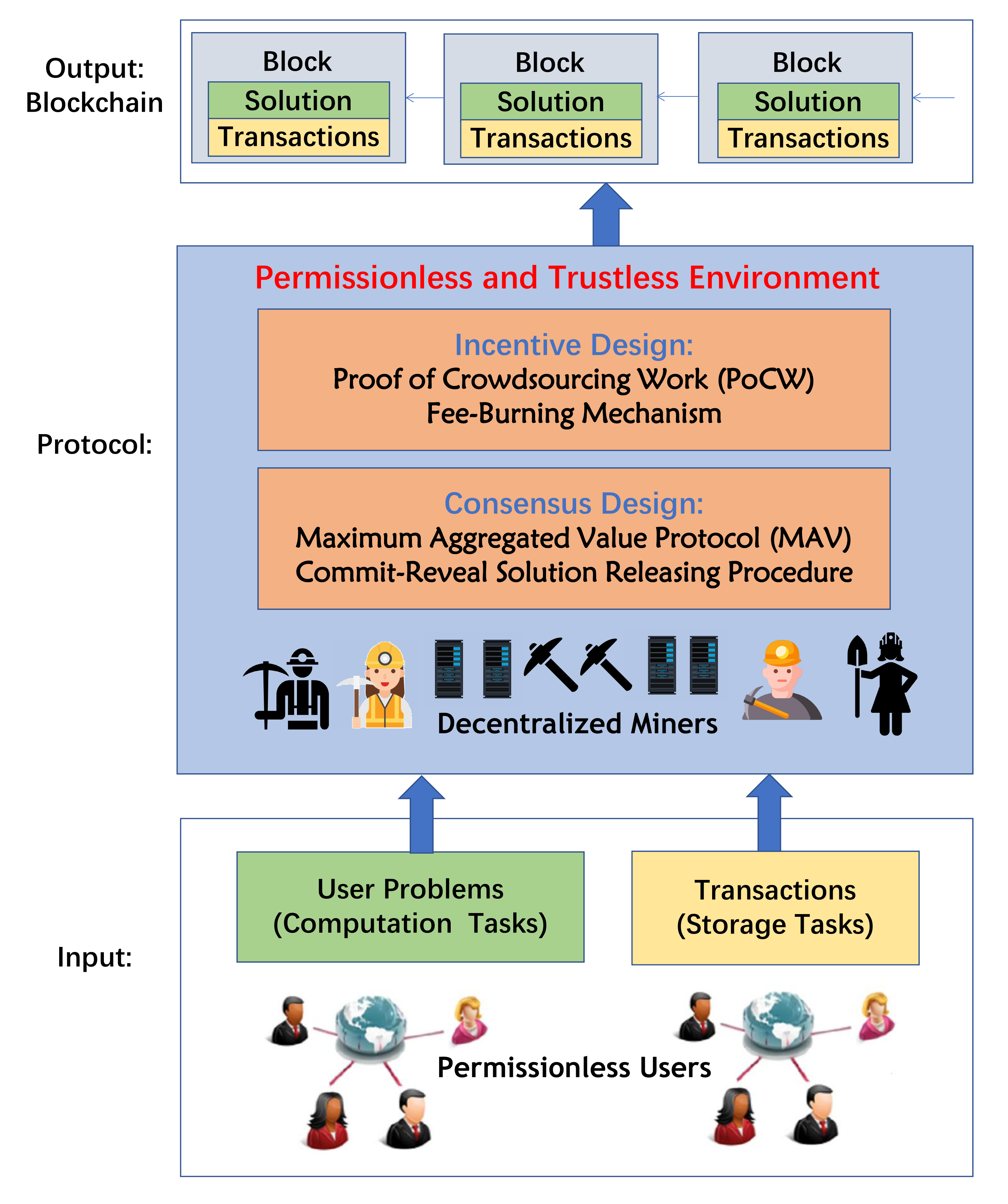}
	\caption{System architecture of CrowdMine.} 
	\label{fig:system_model}
\end{figure}

\subsection{Participants}
The system is permissionless and anonymous. Similar to many account-based blockchains like Ethereum, each participant in the system is represented by an address associated with its asset account.
There are two different roles of participants in the system, i.e.,   miners and users.
Users are participants that need to get their computation tasks solved or to get some transactions (or other types of messages) recorded in the system.
They announce the tasks to be solved or the transactions to be recorded, as well as the associated reward. 
Miners are the block creators in the system. They earn the reward by solving computation tasks or recording transactions requested by users.

\subsection{The Blockchain}

\subsubsection{Transactions}
Transactions are the records that the system stores on chain, each of which records a monetary transfer.
The sender of a transaction must include the following information in the transaction: the sender account, the receiver account, and the transaction value (i.e., the amount of transferred asset). 
In addition, the sender proposes the transaction fee to reward the miner who includes the transaction in the newly-mined block on chain.

\subsubsection{Problems}
Problems are computation tasks in the system.
There are two types of computing tasks, the user problem and the system problem. 
User problems are computation tasks proposed by users, where the problem description and the reward list are both specified (see Section \ref{sec:miningpro} for the detailed protocol).
System problems are hash puzzles similar to traditional PoW, where miners need to find a nonce that satisfies some predetermined and publicly-known cryptographic constraint \cite{nakamoto2008bitcoin}. 
System problems are available to miners all the time, so as to ensure timely publication of transactions when there is no available user problem. The problem description and reward for system problems is pre-determined and stable.

\subsubsection{Blocks}
Only when a miner successfully solves a problem, a new block can be created and appended to the end of blockchain.
The miner can claim the problem reward in the block, and include transactions in the block to obtain transaction fee. 
Thus, a block is associated specifically with one solved problem, and contains multiple transactions.
The detailed block generation process and chain rule will be introduced in Section \ref{sec:miningpro}.

\section{Mining Process Details}\label{sec:miningpro}
In this section, we introduce the mining process in \emph{CrowdMine}.
Figure \ref{fig:problem_process}  demonstrates the workflow for users and miners.

\begin{figure}[!th]
	\centering
	\includegraphics[width=1\linewidth]{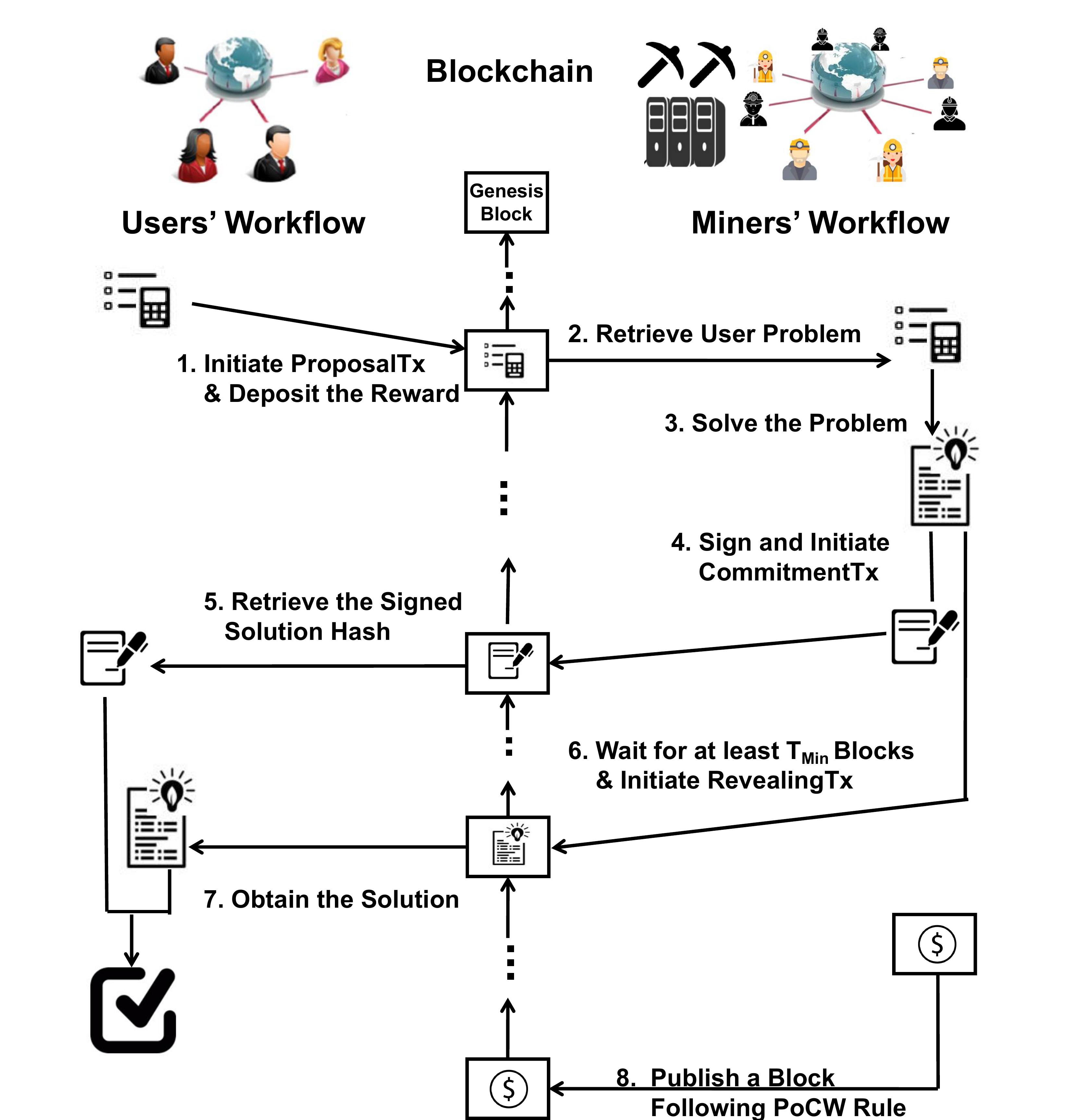}
	\caption{Problem handling process.} 
	\label{fig:problem_process}
\end{figure}

\subsection{Proposing a User Problem}

To propose a user problem, a user needs to initiate a transaction called \emph{Problem Proposal Transaction} ({\tt{ProposalTx}} for brevity). 
The transfer amount of {\tt{ProposalTx}} is the highest problem reward and should be deposited in the system.
Since many real-world computation tasks accept solutions of different quality, in addition to the highest reward,  a publicly verifiable criterion for acceptable solutions should also be specified, along with a reward list $\theta_1, \theta_2,....\theta_{l-1},\theta_{l}> 0$ to all possible legitimate solutions classified by users in $l$ levels, from the highest quality level to the lowest level (i.e., ``solution-not-found"). To avoid spamming of unsolvable tasks, we introduce the minimum reward constraint. Even for the result ``solution-not-found", users still need to pay at least a portion $\xi>0$ of the deposited solution reward, so as to compensate for miners' commitment in trying to find a  solution, where $\xi$ is a system parameter.

To encourage high-quality solutions, we also introduce the timeout mechanism. A search time $T_{\rm search}$ is associated with each task. Within $T_{\rm search}$ blocks after the problem proposal, only a solution of the highest quality specified by the reward table can be accepted. Imperfect solutions eligible for lower rewards can only be accepted only after $T_{\rm search}$ blocks have been appended on the main chain after the corresponding {\tt{ProposalTx}}.

The minimum reward constraint and the timeout mechanism, are introduced to prevent our system against the potential risk of denial-of-service attack, where unsolvable tasks are deliberately proposed for spamming at rather low cost and lead to congestion in the system.

After {\tt{ProposalTx}} is published on chain, the corresponding assets should be locked from the problem proposer. 
There are two possible ways to unlock and retrieve the assets (note that they can co-exist for one task). One is by a miner with a proper solution to claim the reward, while the other is by the task proposer to retrieve the unspent part of the reward, if miners fail to find a solution of highest quality. This mechanism is similar to the Hash Time Lock Contract (HTLC) in the lightning network \cite{poon2016bitcoin}.

\subsection{Reward-Burning Scheme}\label{burnt-fee}

For a confirmed block, the miner will receive the block reward, which is composed of the problem reward and the transaction fee. 

\textbf{Reward-burning Mechanism.} Here, we introduce a ``burning'' mechanism on the problem reward. Let $R_{\rm problem}$ denote the final amount of problem reward paid by the problem proposer according to the reward list. Note that we can also represent the reward for a system problem as $R_{\rm problem}=\tilde{R}$, where $\tilde{R}$ is a constant configured by the system. The miner will receive only part of the whole reward $R_{\rm problem}$, i.e., only $(1-k) R_{\rm problem}$ is transferred to the miner, while the remaining $k R_{\rm problem}$ of the problem reward can be considered as having ``burnt'' by the system. Here $k\in[0,1)$ is the reward-burning ratio decided by the system. Note that, the reward-burning mechanism only affects the problem reward. For the transaction fee, all the amount paid by the users will be received by the miner.

The reward-burning mechanism is similar to the EIP-1559 protocol in Ethereum \cite{buterin2019eip}. Such a mechanism not only helps stabilize the circulating tokens (see Section \ref{sec:burn} for detailed discussions), but also addresses security concerns that will be elaborated in Section \ref{sec:security}.

\subsection{Miner's Work Flow}\label{minerswork}
During the mining process, 
miners can either pick up one user problem to solve, or solve a system problem.
After successfully finding a solution, the miner starts the block generation process as follows.

If the miner solves a system problem, similar to traditional PoW-based systems, the miner can broadcast a block which contains the solution and claim the system problem reward $\tilde{R}$ immediately, and other miners can easily verify the solution.


If the miner solves a user problem, to prevent the solution from being stolen, the miner should follow the \textbf{Commit-Reveal Solution Releasing Procedure} to safely broadcast the solution, claim the reward and generate a block.

First, the miner should initiate a \emph{Solution Commitment Transaction} ({\tt{CommitmentTx}} for brevity), which is a transaction that stores the hash digest of the miner's account address and the solution.
Note that, with the cryptographic guarantee, other miners cannot obtain the solution by observing  {\tt{CommitmentTx}}.

After the confirmation of {\tt{CommitmentTx}}, the miner should further initiate a \emph{Solution Revealing Transaction} ({\tt{RevealingTx}} for brevity) which records the solution in consistency with the hash digest.
 {\tt{RevealingTx}} is allowed to be published on chain only within a certain time window after the block containing the  commitment transaction. Specifically,  {\tt{RevealingTx}} should be published $t \in [T_{\rm min}, T_{\rm max}]$ blocks after the corresponding {\tt{CommitmentTx}}, where both $T_{\rm min}, T_{\rm max}$ are both system configurations.
After $T_{\rm max}$ blocks, if no corresponding {\tt{RevealingTx}} appears on chain, the {\tt{CommitmentTx}} is expired, and other miners are allowed to propose new commitments.
Recall that when there are multiple solutions are committed and revealed on chain before the problem timeout $T_{\text{search}}$, only the solutions with the highest quality are valid. 
Particularly, when there are multiple solutions with the highest quality, the system can specify a tie-breaking rule, e.g., by the order of the corresponding {\tt{CommitmentTx}} on chain. 

\textbf{Proof of Crowdsourcing Work (PoCW).} 
After successfully publishing {\tt{RevealingTx}}, the miner is eligible to pack transactions and publish a candidate block following the Proof of Crowdsourcing Work (PoCW) protocol.
Specifically, the miner can create a valid block if the candidate block satisfies the condition of PoCW:
\begin{equation}\label{eq:pos_block}
	\text{Hash}(\text{{\tt{time}}, candidate block}) < D \cdot R_{\rm problem}.
\end{equation}

In (\ref{eq:pos_block}), {\tt{time}} represents the timestamp when the candidate block is generated, and $D$ is a pre-determined mining difficulty. 

The novelty of PoCW is that the probability of generating a new block with a user problem depends on the reward of the solved problem, $R_{\rm problem}$. It takes no previous accumulated stake (as in PoS), nor meaningless hash computation (as in PoW).
Specifically, the problem reward $R_{\rm problem}$ fairly reflects the miner's recent contribution/workload in the crowdsourcing process. The greater contribution that a miner makes, the quicker the miner can generate a new block. Since the problem reward is decided by users in the market, PoCW transfers the value of miners' useful work into the block generation right by measuring the problem reward.



\subsection{Chain Rule: Maximum Aggregated Value}\label{sec:mode:chain}

In our system, we maintain a linear and sequential chain of blocks, i.e., the ``main chain", for consensus.
For orphaned blocks that are not recognized in the main chain, the miners will not receive any reward, and the included transactions won't be admitted, either.

When miners see forking branches of chains, \textbf{ ``Maximum Aggregated Value Protocol'' (MAV)} is applied as the rule for the main chain selection. Define the ``value'' of a block to be the reward of the associated solved problem, regardless of whether the problem is a user-proposed one or a system-generated one.
According to the MAV protocol, the branch with the maximum aggregated block value should be selected as the main chain. 
If multiple forks have the same aggregated value, a miner can arbitrarily select one as the main chain to continue mining, and the tie is broken once another block is appended to one of the forks. 

\begin{figure}[!t]
	\centering
	\subfigure{
		\begin{minipage}{0.95\linewidth}
			\centering
			\includegraphics[width=3in]{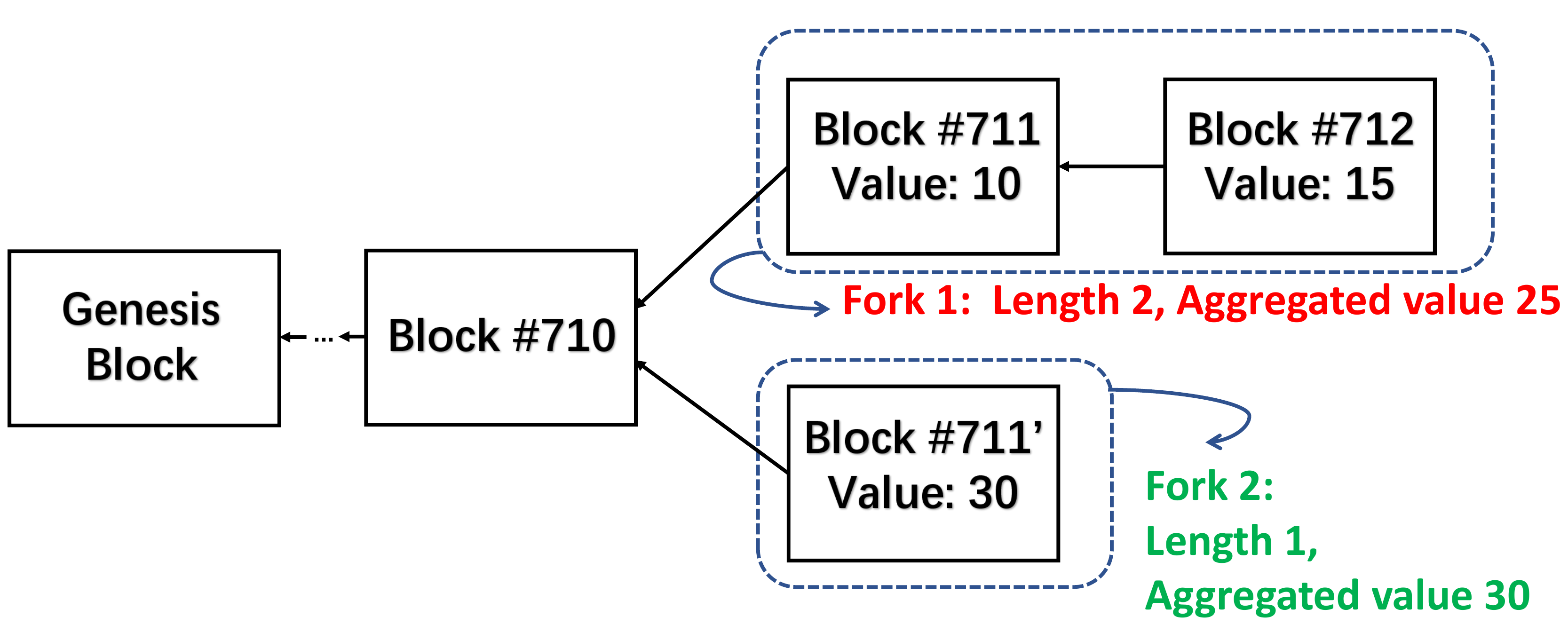}
		\end{minipage}
	}
	
%
	
	\caption{Fork 2 has a larger aggregated block value.}
	\label{fig:fork_max}
\end{figure}

Figure \ref{fig:fork_max} shows an example of two fork branches, where the branch with blocks \#711, \#712 has an aggregated value of $25$ and the branch with block \#711' has an aggregated value of $30$. In this case,  the branch with block \#711' should be the main chain.

\section{Security Analysis}\label{sec:security}
In this section, we discuss the system security against malicious users and malicious miners. We first focus on possible threats due to the unique design and feature of the proposed system. Then, we also examine the system's security guarantees on the classic 51\% attack.

\subsection{Defend Against Malicious Users}

 Compared with classical blockchain systems where users can only initiate transactions, a unique feature in our system is that users can propose useful computation problems. 
 Thus, it is important to address the possible malicious user behavior during problem proposal.
 Since a user can create multiple identities/addresses in a decentralized system, a malicious user is able to exploit the protocol by constructing problems as follows. The attacker can propose a problem with a known solution. Then, the attack can solve the problem with another identity to win the permission to produce a block and earn extra transaction fee. Moreover, the attack may lead to severe congestion in the network, if a lot of users behave this way.

To address this issue, for miners to pack transactions in a block, the following constraint should apply.
\begin{constraint} \label{tx-volume}
	(\textbf{Transaction Volume Constraint}) For a block with  problem reward $R_{\rm problem}$, the total transaction amount of all included transactions, $V$, should be less than $k \cdot R_{\rm problem}$, where $k\in [0,1)$ is the reward-burning ratio.  
\end{constraint}

    The transaction volume constraint implies that, a block with a higher problem reward can include more transactions in value. 
    In the following, we discuss how this constraint prevents malicious users from exploiting the permissionless problem proposal to revert transactions or to grab transaction fees without useful computation.

	\subsubsection{Double-spending} Double spending occurs when the miner can revert an on-chain transaction, by forcing other miners to switch to a deliberately forked chain with a conflict transaction. In this case, the attacker can spend the money twice. 
	We consider the attacker tries to revert a transaction $tx$, by building a private chain from solving the attacker's problems. Let $u_{\rm double-spending}(tx)$ denote the attacker's payoff of double spending $tx$ by forking a new branch of the constructed attacker problems. We have the following result.
	

	\begin{theorem}\label{thm:constructing}
	    Double spending a transaction $tx$ by constructing attacker problems is not profitable, i.e., $u_{\rm double-spending}(tx)<0$.
	\end{theorem}
	
	\begin{proof}
	
		We first consider the case where the target transaction $tx$ is packed in the newly-mined block. Let $R_{\rm problem}$ and $V$ denote the actual problem reward and total transaction value of the newly-mined block, respectively. 
		In this case, the attacker only needs to construct one problem to revert the transaction. Denote the attacker problem reward as $R_{\rm attacker}$.

    	We know that the value of transaction $tx$, i.e.,  $v_{\rm tx}$ is less than $k R_{\rm problem}$, since we have  $v_{\rm tx}\leq V \leq k R_{\rm problem}$ according to the transaction volume constraint.
    	Also, according to the MAV chain rule introduced in Section \ref{sec:mode:chain}, if the attacker wants to fork the main chain and revert the transaction, the aggregated value of the attacker's new branch should be larger than the current main chain, i.e., $R_{\rm attacker} \geq R_{\rm problem}$. 
		
    	Due to the fee-burning mechanism, it at least costs the attacker $kR_{\rm attacker}$ to solve the constructed problem, and the gained value will be $v_{\rm tx}$ if the attack succeeds.
    	Thus, the attacker's profit $u_{\rm double-spending}$ by successfully conducting the double-spending attack is at most
    	\begin{align*}
    	  	u_{\rm double-spending} =& v_{\rm tx} - k R_{\rm attacker}\\
    		\leq & V - k R_{\rm problem} < 0.
    	\end{align*}
    	
    	Note that if the attacker wants to revert an older transaction included in a previous block, we can similarly show that it costs the attacker even more money to construct enough blocks. Thus, double spending by constructing attacker problems is not profitable.
	\end{proof}

	The result in Theorem \ref{thm:constructing} shows that the system can naturally prevent malicious users from gaining extra profit through double-spending in the block associated with the constructed problem.

	\textbf{Remark.} (\emph{Value of on-chain data}) Consider a special case where a user proposes a transaction that sends money from and to the same account. 
	In this case, the purpose of the transaction is usually to store (on chain) the messages/data that are appended to the transaction.
	For such a data storage request, the user still needs to decide the value of the transaction. Though the money still transfers back to the same user, a higher transaction value commits a higher transaction fee, since it occupies more ``space'' in the block, due to the transaction volume constraint.
	Thus, when deciding the value of the transaction, the user should carefully evaluate the value of the appended data to herself/himself, as well as the value to potential attackers. 
	This is because the transaction may be at risk of being double spent if the claimed transaction value is too low and fails to represent the value of the stored data.
	From this point of view, the data stored on chain will be properly evaluated by the corresponding transaction value by concerned users.
	Moreover, the data is protected by the blockchain as an insured object with the value of the transaction amount, and the transaction fee is the insurance fee for the data storage.

	
    
    
    	


	\subsubsection{Fee grabbing}  
	
	Other than double spending a transaction, a malicious user may also attempt to earn extra transaction fee.
	Consider that a malicious user constructs a user problem to grab the transaction fee from the generation of the next block, and denote $u_{\rm fee-grabbing}$ as the attacker's payoff from this fee grabbing attack. We have the following result.
	
	\begin{theorem}
          An attacker's payoff from the fee grabbing attack by constructing a problem with known solution is negative, i.e., $u_{\rm fee-grabbing}<0$.
    \end{theorem}
    
    \begin{proof}
	Note that the transaction fee is less than the corresponding transaction value. Assume that the attacker constructs an attacker problem with problem reward $R_{\rm attacker}$. According to the transaction volume constraint, the total transaction value in this block is lower than $kR_{\rm attacker}$, so is the total transaction fee $F$, since $F\leq V$. The profit for fee grabbing is 	
	\begin{equation*}
	\hspace{-0.2in}	u_{\rm fee-grabbing} = F - k R_{\rm problem} \leq V - k R_{\rm problem} < 0.
	\end{equation*}
    Thus, it is not profitable for the attacker to gain extra profit from fee grabbing by constructing a user problem with known solution.
    \end{proof}

\textbf{Remark:} The transaction volume constraint is a strict and rigorous mathematical guarantee for security and may lead to low throughput of our system. In practice, we can adjust the constraints for more plausible throughput without triggering the security issues in different ways. Here are some possible alternatives. For example, similar to the six-block-confirmation mechanism in Bitcoin, the user can confirm the included transaction only after the aggregated value is greater than some specified threshold, e.g., $\frac{2V}{k}$, which can increase the cost of the double-spending attack. 
Besides, adding a block-finalizing time in consensus could help 
prevent the fork caused by the double-spending attack.

\subsection{Defend Against Malicious Miners}

Next, we investigate the malicious behaviors of miners.
The major difference between our proposed system and classic PoW systems is the proof for a miner to generate a block. In PoW, the proof is a number (nonce) that is valid only on the corresponding miner address. But in \name, any miner that provides a solution to the user problem is eligible for block generation. 
Thus, a natural concern is that what happens if a malicious miner copies the published solution and generate blocks based the stolen solution?

Note that that once a {\tt{RevealingTx}} is broadcast, the attacker can observe the solution and steal it. The attacker can copy and claim the solution on a deliberately created fork after the solution is revealed. If the malicious fork turns out to be the main chain, the attacker steals the solution. 

	\begin{figure}[!ht]
		\centering
		\includegraphics[width=3.3in]{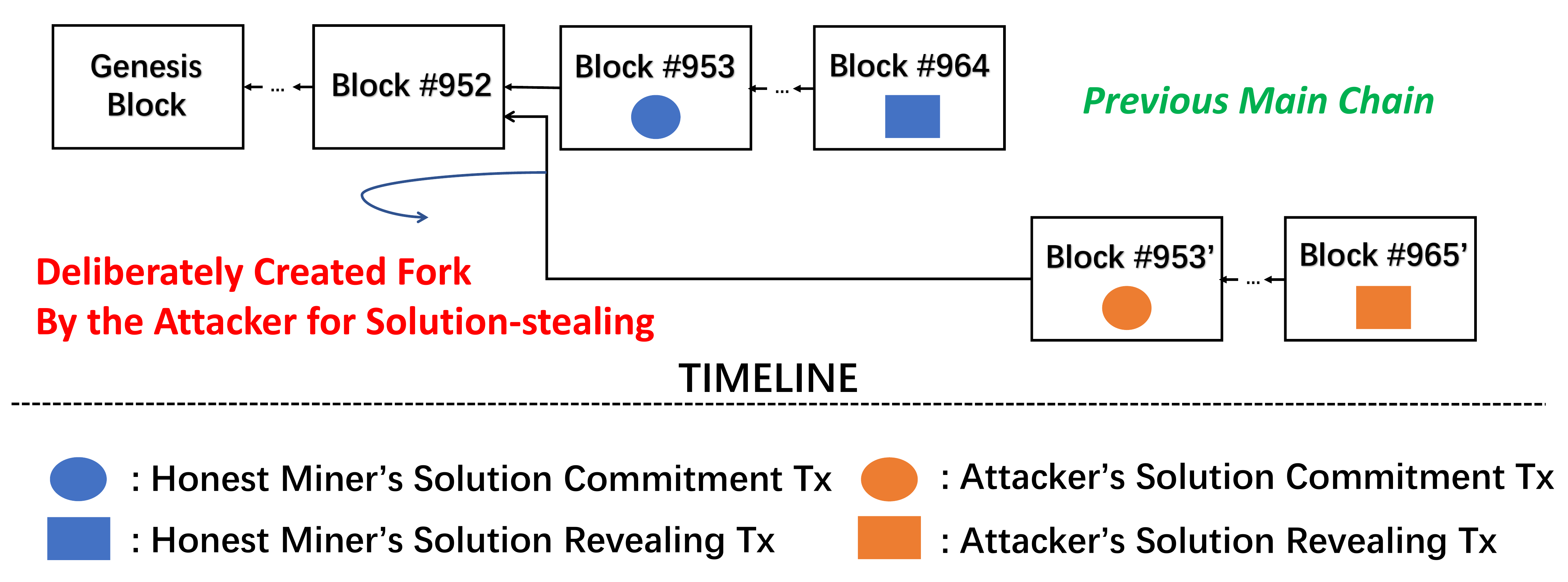}
		\caption{Demonstration of the solution-stealing attack.} 
		\label{fig:solution_stealing}
	\end{figure}

	Figure \ref{fig:solution_stealing} demonstrates the solution-stealing attack with $T_{\rm min}=10$. On the main chain, an honest miner solves a user problem and initiates the corresponding {\tt{CommitmentTx}} which is packed in Block \#953. After waiting for $T_{\rm min}=10$ blocks, the miner initiates the corresponding {\tt{RevealingTx}} which is packed in Block \#964. 
	Afterwards, the attacker observes the solution from the corresponding {\tt{RevealingTx}} and tries to steal it as follows. The attacker deliberately creates a fork starting from Block \#953', following the Block \#952. Block \#953' contains the {\tt{CommitmentTx}} associated with the address of the attacker. Then, after managing to append $T_{\rm min}+1=11$ blocks on the attacker's branch, the attacker can reveal the solution in a transaction contained in Block \#965'. In this case, the deliberately created fork becomes the main chain and the solution is successfully stolen by the attacker.
	
	However, thanks to the \textbf{Commit-Reveal Solution Releasing Procedure}, our system is resistant to the attack.
    Note that the attacker can propose a {\tt{CommitmentTx}} on the fork only after observing the {\tt{RevealingTx}} on the main chain. We know that {\tt{RevealingTx}} is at least $T_{min}$ blocks after the corresponding {\tt{CommitmentTx}}. Thus, when the attacker tries to steal the solution, the {\tt{CommitmentTx}} on the fork is at least $T_{\rm min}$ blocks behind the main chain.
    Then, the attack can succeed only if the attacker can generate a fork which is at least $T_{\rm min} + 1$ blocks longer than the main chain in a very short period.

    To generate these blocks on the attacker's private chain, the attacker needs to solve $T_{\rm min} + 1$ problems, each being either a system problem or a user problem. If the attacker generates blocks by solving system problems, the security threat is the same as launching a 51\% attack in the system, which will be discussed in Section \ref{sec:51attack}. On the other hand, if the attacker generates blocks by solving user problems or even constructing user problems, the block generation process will be delayed by the PoCW protocol, which traps the attacker at a disadvantage in the race between the main chain and the attacker's chain.  In this way, the probability that the attacker's chain catches up with the current main chain with $T_{\rm min}$ blocks in the system drastically decays with the value of $T_{\rm min}$, which means the attack hardly succeeds. This is similar to the six-block-confirmation mechanism of Bitcoin. Furthermore, note that the solutions of system problems on the main chain can never be stolen because they are associated with the solver's address. Since these solutions also contribute to the aggregated value of the main chain, they also enhance the system's capability against the solution-stealing attack.

    A larger value of $T_{\rm min}$  provides a higher level of security, but affects the users' quality of experiences, as it also indicates an increase in the latency of the solution revealing procedure. 
    Therefore, depending on the required security level, there is a subtle trade-off between security and users' satisfaction in the design of the system parameter $T_{\rm min}$.

\subsection{Better Resistance to the 51\% Attack}\label{sec:51attack}

	The 51\% attack indicates a situation when a miner possesses more than half of the computational resource in the system. An attacker with more than 50\% computing power can conduct the double-spending attack by solving problems and reverting transactions. The 51\% attack is mostly considered to be hypothetical since such situation is somewhat unrealistic and unstable in the long term \cite{aponte202151}. However, in many cryptocurrencies, it only costs the attacker thousands of US dollars to buy computing services and boost one’s computing power to more than half of the system for a short period of time \cite{crypto51}.
	This implies that once a high-value transaction is processed in these cryptocurrencies, conducting an immediate 51\% attack to double spend the transaction could be profitable.

    Compared with most PoW-based cryptocurrencies, we argue that our system is more robust against the 51\% attack, in that our system can defend the short-term 51\% attack. The key reasons are two folds.
    
    The first reason is that when the miner solves a user problem, the block will not be published immediately according to the PoCW rule. In traditional PoW, attackers can buy computing power and generate numbers of new blocks in a short time to fork the main chain. But in PoCW, even if the attacker buys computing power to solve some user problems in a short time, the blocks cannot be published immediately, as it is required to wait for PoCW block generation process. While the attacker is waiting for the block generation, the remaining honest miners are still mining and aggregating value on the main chain, offsetting the short time computation advantage of the attacker.
	
	The second reason is that in our system, if the attacker attempts to revert a high-value transaction, it is mandatory to solve a problem with a high reward, i.e., the problem that requires more computational resources. 
	This is different from traditional PoW where the computational resource needed to mine a block is relatively fixed and irrelevant to the value of the included transactions. 
	Therefore, one is not likely to revert a high-value transaction with a short-term boost in computational resources.

\section{System Implementation }\label{sec:experiment}

\begin{figure*}[!th]
	\centering
	\subfigure[Block generation rate.]{
	\includegraphics[width=0.30\linewidth]{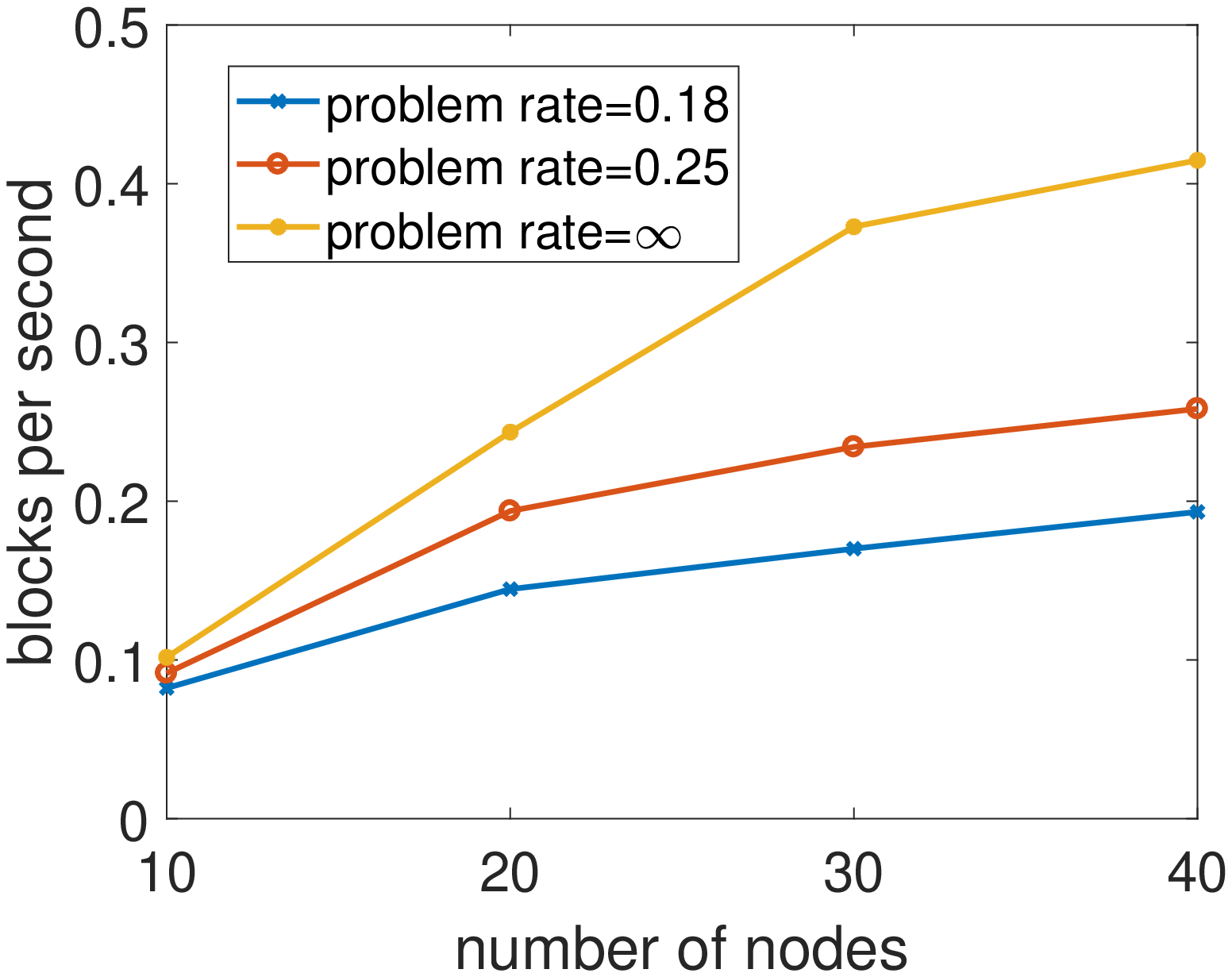}
	\label{fig:block_rate}}
	\centering
	\subfigure[Problem process rate. ]{\includegraphics[width=0.30\linewidth]{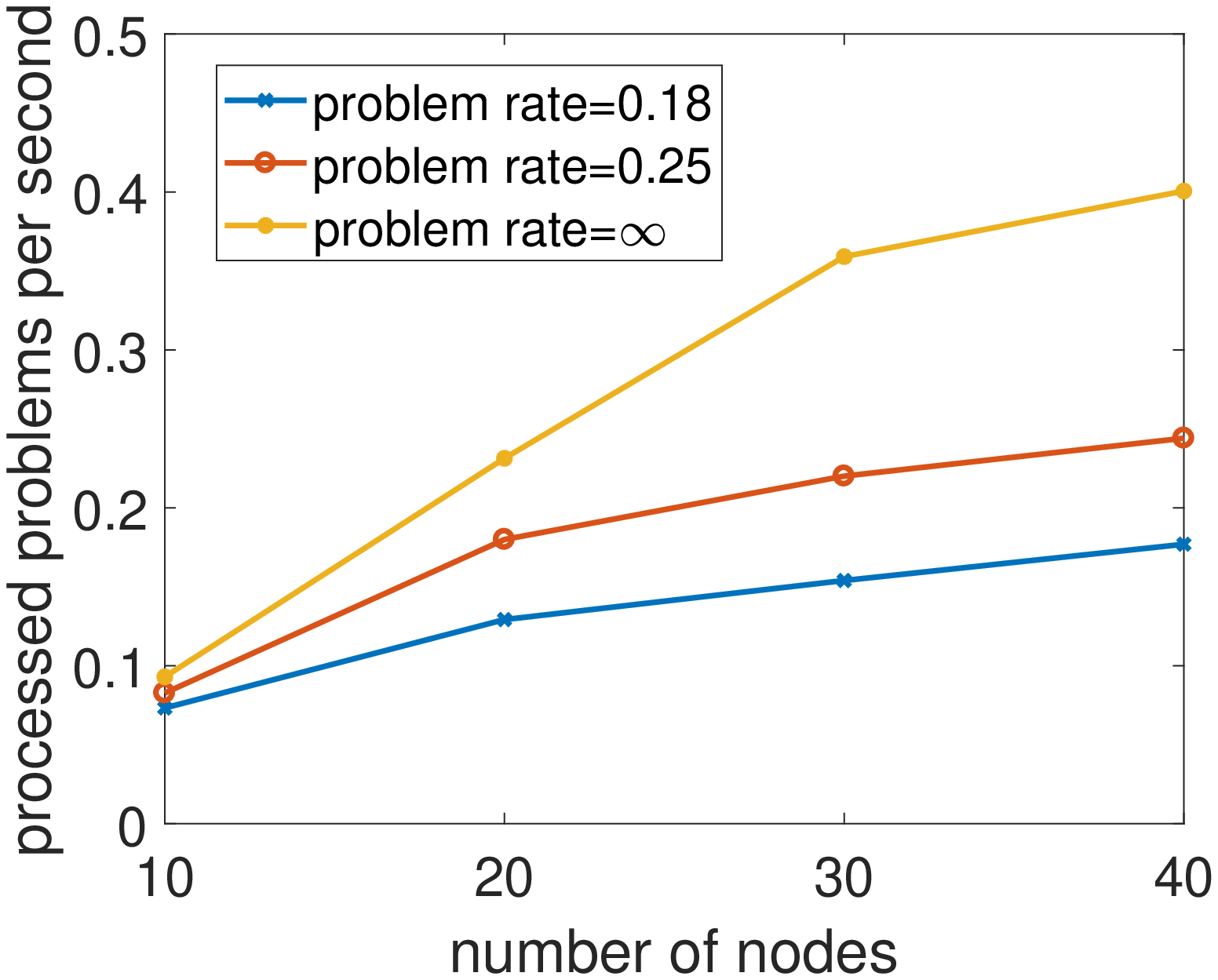}\label{fig:pb_rate}}
	\hfil
	\centering
	\subfigure[Transaction process rate.]{
		\includegraphics[width=0.30\linewidth]{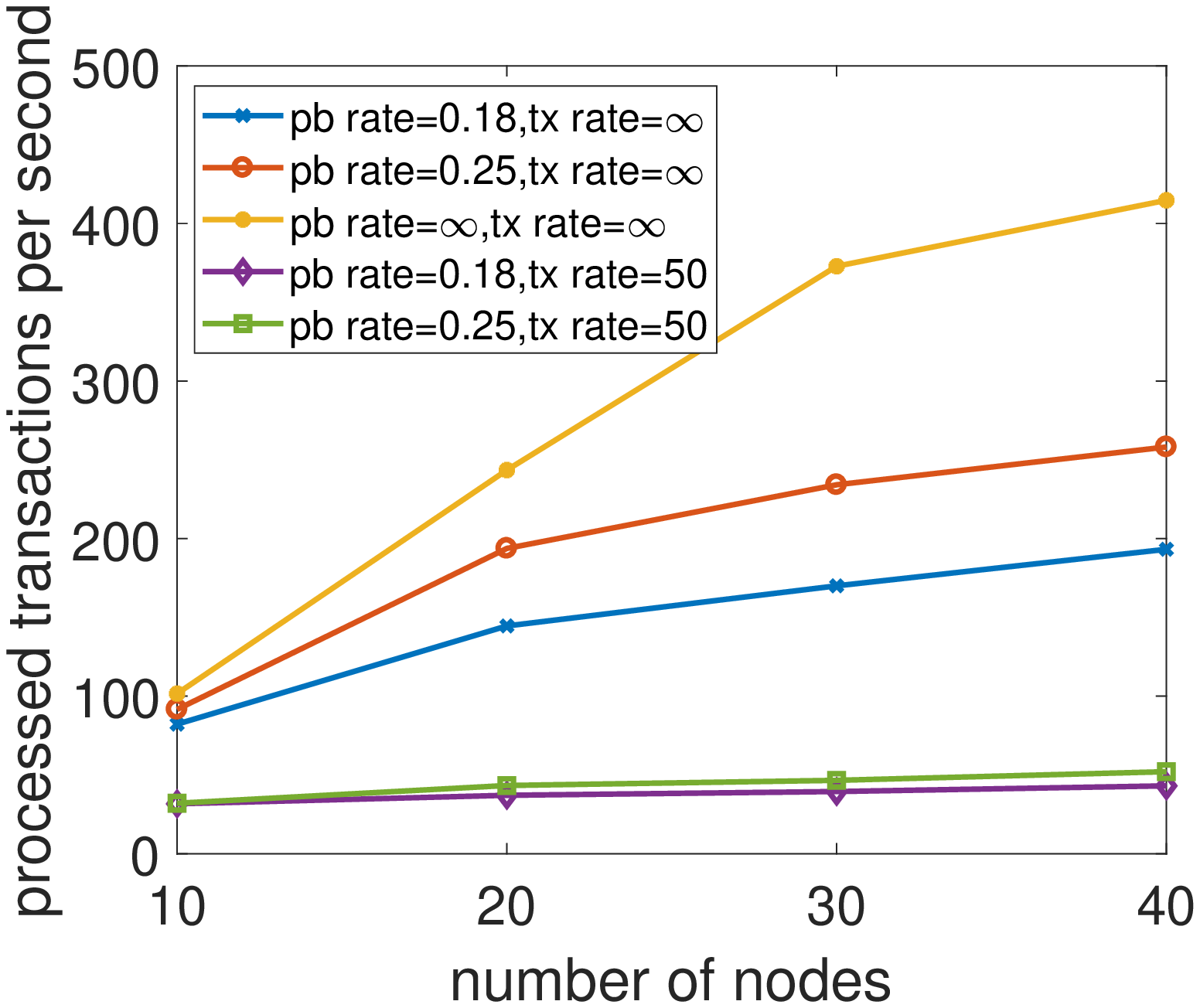}\label{fig:tx_rate}}

	\subfigure[Effective computational resource utilization. ]{\includegraphics[width=0.30\linewidth]{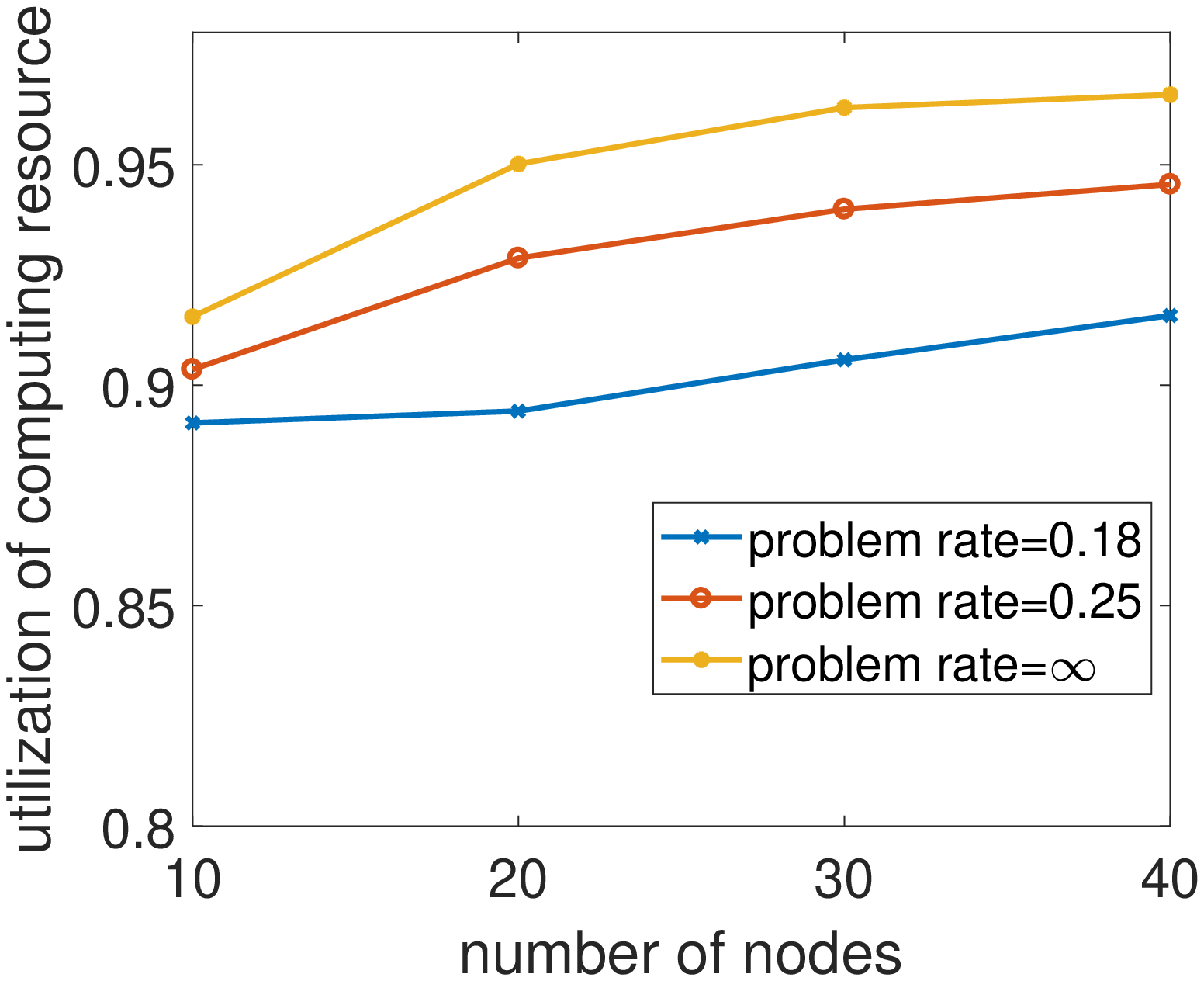}
	\label{fig:utilize}}
	\centering
	\subfigure[Circulating supply of token. ]{
		\includegraphics[width=0.30\linewidth]{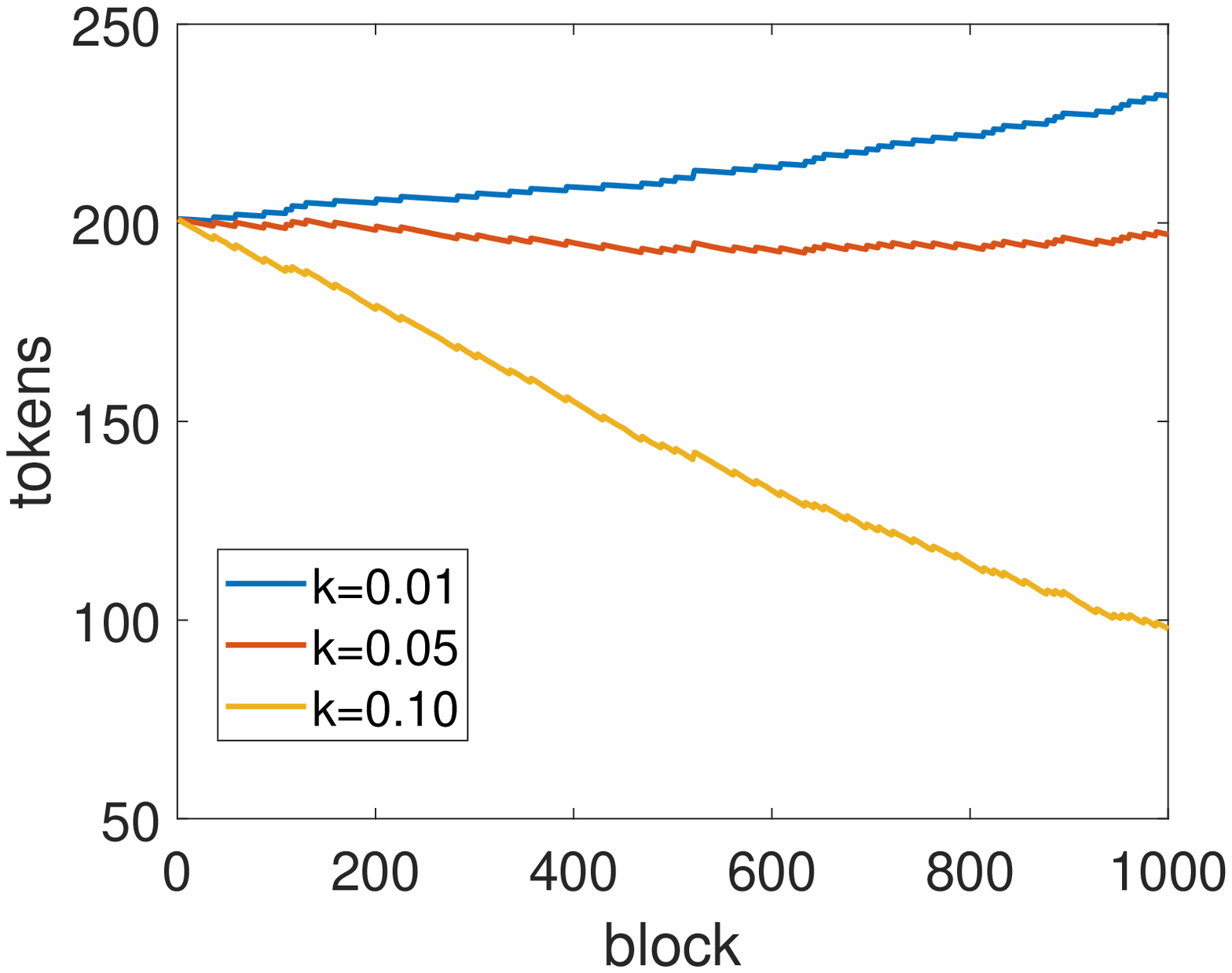}
		\label{fig:circulate}}
	\hfil
	\centering
	\subfigure[Zoom-in view of (e) on the first 100 blocks.] {\includegraphics[width=0.30\linewidth]{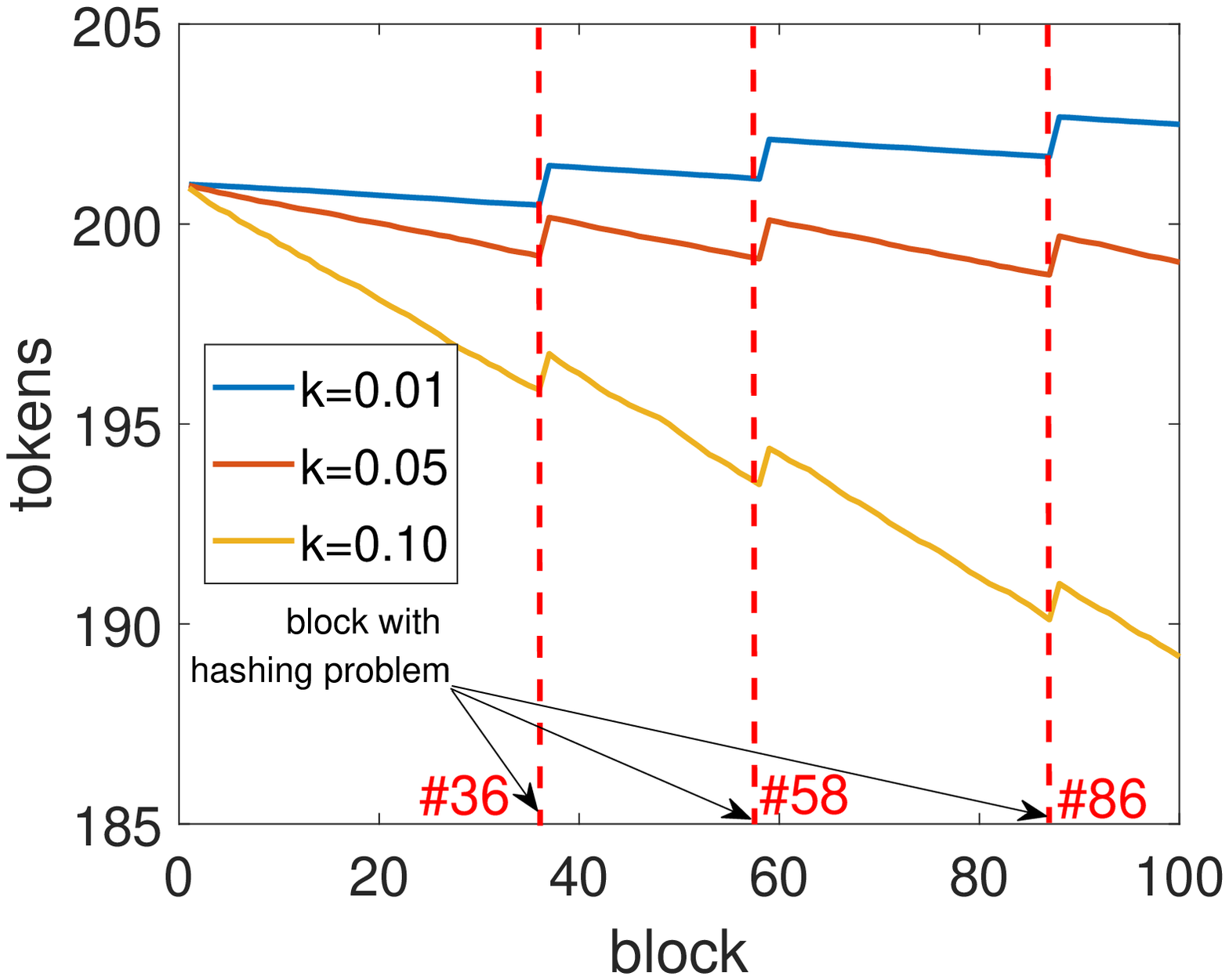}	\label{fig:zoom}}
	
	\vspace{-0.15in}
	\caption{Results of experiment.}
	\label{fig:exp:results}
	\vspace{-0.10in}
\end{figure*}

\subsection{Resource and Environment Configuration}

We have built a decentralized prototype to demonstrate the performance and properties of our proposed system.
The system is built in Golang, and we use database BoltDB (v1.3.1) to store blocks and transactions. 
All our experiments are performed on a server with one AMD Ryzen 5950X (16 cores, 32HT) CPU and 32GB of memory, with up to 40 distributed docker containers running Ubuntu 20.04.3 LTS are deployed to simulate the decentralized system .

%

In the experiments, we adopt the star network to organize the decentralized nodes (i.e., containers), where every node connects to the hub of the network, i.e., a router.
The router in the star network routes transactions, problems, and blocks among these decentralized nodes. Such structure allows the nodes to simulate the peer-to-peer communication. The router also collects the data in the network to form our experimental results.

\subsection{Implementation Details}
In this section, we introduce the detailed implementation.

\textbf{Problem Paradigm: Constraint Satisfaction Problem (CSP).}
In this paper, we implement the classic Constraint Satisfaction Problem (CSP), a general optimization framework that captures many practical scenarios \cite{prosser1993hybrid}, as a demonstration of user problems. CSP allows solutions of different qualities, by satisfying different numbers of constraints.

\textbf{Problems.} Three different types of user problems are implemented in our experiments: graph coloring \cite{jensen2011graph}, Sudoku \cite{simonis2005sudoku}, and zero-one programming \cite{padberg1975note}. These problems are classic representatives of CSPs. We adopt the exhaustive random search, a general technique \cite{heule2017science}, to solve these problems. Besides, the problem difficulty is adjusted so that each problem can be solved within around 10 seconds with 10 decentralized nodes in the network.

\textbf{Users.}  We simulate users' problem proposals by letting the router generate and broadcast user problems within the star network. 
The router will randomly generate the aforementioned three kinds of problems with a ratio of 1:1:1 to simulate heterogeneous user demand. 
Meanwhile, the router also generates and broadcasts transactions to simulate the transaction circulation process.

\textbf{Miners.}  
 Each miner independently runs on a docker container, with predetermined exclusive computing resources. Miners communicate with each other within the star network through the router.  Each miner randomly and independently selects user problems. In our experiment, miners are honest and strictly follow the chain rule. 

\subsection{Performance Results}
In this section, we investigate the performance of the proposed system from two crucial perspectives, namely, the system throughput and the resource utilization.

\textbf{Throughput.}
We conduct experiments with different network sizes under different problem generation rates and transaction generation rates, and the results are shown in Figure \ref{fig:exp:results}. 

Figure \ref{fig:block_rate} demonstrates the system performance with the metric of block generation rate measured by the number of blocks generated on the main chain per second.
It shows that, with more miners and more computation power, the block generation rate will monotonically increase, and it implies good parallelism of our system. 
Figure \ref{fig:pb_rate} shows the problem process rate, which is defined as the number of user problems being solved per second. 
When the problem generation rate is low, the system performance is limited by the problem generation rate. Therefore, the number of processed user problems per second increases with problem generation rate.
Figure \ref{fig:tx_rate} shows the transaction process rate, 
which is defined as the number of transactions published on the main chain per second. 
It shows that, though both problem generation rate and transaction generation rate affect the transaction processing rate (TPS), the dominant factor affecting TPS is the former one. Also, it demonstrates plausible scalability and throughput of the system.

\textbf{Resource utilization.}
Define the computing resource utility as the proportion of computing power for solving user problems.
Figure \ref{fig:utilize} shows the computing resource utilization in the system, from which we can find that the utilization rates increase when more miners join our system, which shows the great scalability and environmental friendliness of the designed system. 
Also, we can find that the utilization increases with problem generation rate. It is because  when the problem generation rate increases, miners can solve more user problems instead of meaningless system problems. 



\subsection{Inspecting Circulating Supply of Token}\label{sec:burn}

In Section \ref{sec:security}, we have shown that the reward-burning mechanism is crucial for the system security. In this section, we further investigate the reward-burning mechanism and analyze its impact on the circulating supply of the cryptocurrency tokens in the system.

We record and plot the total tokens in the system during the generation of the first 1000 blocks in our experiment with 40 nodes and problem generation rate 0.2. The results are shown in Figure \ref{fig:circulate} and Figure \ref{fig:zoom}.
As shown in Figure \ref{fig:circulate}, we can tell that, when $k$ is smaller than 0.05, the supply of tokens in the system increases.
When $k$ is larger than 0.05, the token supply decreases.
Specially, at the tipping point $k=0.05$, tokens remain relatively stable.
Figure \ref{fig:zoom} shows the zoom-in view of the first 100 blocks in Figure \ref{fig:circulate}.
From Figure \ref{fig:zoom}, we can find that after several blocks with user problems, there will be a block with hashing problem that includes some solution commitment transactions. The block with hashing problem occurs when there are no unsolved user problems, which happens occasionally.
As mentioned above,  the total cryptocurrency in circulation decreases when miners mine a block with a user problem, and the total cryptocurrency in circulation increases when miners mine a block with a hashing problem. 
Therefore, from Figure \ref{fig:zoom},  at the blocks with system problems, i.e.,  block \#36, block \#58 and block \# 86,  we can observe token increase, while the remaining blocks with user problems lead to token decrease.

\textbf{Remark (Economic observation).}  The experiment on the different token changes after blocks with user problems and blocks with system problems sheds light on self-adaptive elasticity on toke supply.
When there are insufficient user problems in the system, the cryptocurrency tokens in the system will gradually increase, since miners are solving system problems. Such issuance strategy mimics the inflation of currency, which devaluates the tokens and tends to lower the cost for users to propose a problem.
On the contrary, when there are too many user problems in the system, the token supply decreases, which decreases the token liquidity and suppresses user demand. 
However, the total token supply of the system will not decrease to 0 because the system has its self-balancing mechanism as described as follows. 
When the total tokens in the system decrease, users will reduce their problem fee reward due to the token tightening. However, the problem reward for the system problem remains stable. When the reward for user problems is no longer attractive, miners will prefer to solve the system problems. As analyzed above, the circulating supply of token increases when miners mine a block with a system problem. Thus, the token supply can fluctuate within a feasible range.
In summary, the designed reward-buning mechanism can help regulate the token supply in the system and keep the ecosystem healthy in long term.

\section{Further Discussions}\label{sec:discussion}
In this section, we discuss some possible improvements on the system implementation.

\vspace{-0.1in}
\subsection{Problem Decomposition}
A problem with high computational complexity takes a long time to solve has several concerns.
From the user's perspective, the user who proposes a very difficult problem may need to wait a long time for the solution, which leads to a bad user experience. 
From the miner's perspective, it is risky to solve a ``big'' problem, though the reward is usually high for these kinds of problems, only one miner can eventually claim the reward. 

A practical solution to address this challenge is to let users decompose the complex problem into several sub-problems with low computational complexity. 
Such decomposition is common and practical.
For example, if a user wants to solve a difficult CSP problem, the user can divide the domain of the CSP problem into several partitions, with each partition being the sub-domain of the sub-problem.  Miners in the system can pick out different sub-problems and solve them in parallel. 
Once all the sub-problems are solved, the user can obtain the final solution of the original problem. 
Moreover, the user can decompose the problem into several sub-problems with redundancy to further improve the service quality. Similar techniques are used in the coded machine learning \cite{ding2021optimal}.

\subsection{DAG-based Paradigm}

Directed Acyclic Graph (DAG) allows multiple blocks to be appended simultaneously, which is suitable for our system.
DAG structure supports high concurrency of solution appending, and can greatly boost the system performance.
As discussed in Section \ref{sec:miningpro}, when the temporary fork happens on the blocks that solve different user problems, the orphan blocks can be appended in the main chain, which is similar to the Inclusive protocol \cite{lewenberg2015inclusive}, where the transactions in the orphan blocks can be included in the main chain.

\subsection{Problem Selection Strategy}

Due to the high concurrency and the network delay, the miners usually do not have complete information of the updated problem pool and the other miners' problem selection. Thus, miners may try to solve the same problem concurrently, generating redundant blocks in the system. The collision in problem selection wastes the miner's computing power and may degrade the system performance. 
Similar problems also occur in the current DAG-based blockchain systems, where the miners may include the same transactions in the concurrent blocks \cite{lewenberg2015inclusive}. The technologies that help to resolve the transaction inclusion collision in DAG-based blockchain such as \cite{chen2022tips} can also be adopted here to resolve the problem selection collision in the system.
Besides, the collision will degrade the miner's revenue. Therefore, the rational miners are motivated to take the equilibrium selection strategy to avoid the collision \cite{lewenberg2015inclusive}. 

\section{Conclusion}\label{sec:conclusion}

In this paper, we design a blockchain-based permissionless and decentralized storage and computing paradigm. 
With the proposed Proof of Crowdsourcing Work mechanism, we show the energy efficiency and security of the system.
We also implement the system with up to 40 decentralized nodes, and the experiment results demonstrate the excellent efficiency and performance of our proposed paradigm.


\newpage

\bibliographystyle{ACM-Reference-Format}
\bibliography{reference}

\end{document}